\documentclass[conference]{IEEEtran}
\IEEEoverridecommandlockouts

\usepackage{amsmath, amsthm, amssymb, mathtools}
\usepackage{float,epsfig}
\usepackage{graphicx}
\usepackage{rotating}
\usepackage{subfigure}
\usepackage{caption}
\usepackage{booktabs}    
\usepackage{subcaption}
\usepackage{tkz-graph}
\usepackage{color}
\usepackage{algorithm}
\usepackage{multicol}
\usepackage[utf8]{inputenc}
\usepackage{braket}
\usepackage{subcaption, multirow}
\usepackage{qcircuit}
\usepackage[overload]{empheq}
\usepackage{hyperref}
\usepackage{xcolor}
\usepackage{color,graphicx,float,lscape,enumerate}
\usepackage{algorithm}
\usepackage{svg}

\usepackage[noend]{algpseudocode}
\usepackage{mathtools}
\DeclarePairedDelimiter\ceil{\lceil}{\rceil}

\usepackage{hyperref}

\usepackage{xparse}

\usepackage{lipsum}



\begin{document}

\newtheorem{theorem}{\bf Theorem}[section]
\newtheorem{proposition}[theorem]{\bf Proposition}
\newtheorem{definition}[theorem]{\bf Definition}
\newtheorem{corollary}[theorem]{\bf Corollary}
\newtheorem{example}[theorem]{\bf Example}
\newtheorem{exam}[theorem]{\bf Example}
\newtheorem{remark}[theorem]{\bf Remark}
\newtheorem{lemma}[theorem]{\bf Lemma}
\newcommand{\nrm}[1]{|\!|\!| {#1} |\!|\!|}

\newcommand{\calL}{{\mathcal L}}
\newcommand{\calX}{{\mathcal X}}
\newcommand{\calY}{{\mathcal Y}}
\newcommand{\calZ}{{\mathcal Z}}
\newcommand{\calW}{{\mathcal W}}
\newcommand{\calA}{{\mathcal A}}
\newcommand{\calB}{{\mathcal B}}
\newcommand{\calC}{{\mathcal C}}
\newcommand{\calK}{{\mathcal K}}
\newcommand{\C}{{\mathbb C}}
\newcommand{\Z}{{\mathbb Z}}
\newcommand{\R}{{\mathbb R}}
\renewcommand{\SS}{{\mathbb S}}
\newcommand{\LL}{{\mathbb L}}
\newcommand{\st}{{\star}}
\def\kernel{\mathop{\rm kernel}\nolimits}
\def\sigan{\mathop{\rm span}\nolimits}

\newcommand{\klasse}{{\boldsymbol \Delta}}

\newcommand{\ba}{\begin{array}}
\newcommand{\ea}{\end{array}}
\newcommand{\von}{\vskip 1ex}
\newcommand{\vone}{\vskip 2ex}
\newcommand{\vtwo}{\vskip 4ex}
\newcommand{\dm}[1]{ {\displaystyle{#1} } }

\newcommand{\be}{\begin{equation}}
\newcommand{\ee}{\end{equation}}
\newcommand{\beano}{\begin{eqnarray*}}
\newcommand{\eeano}{\end{eqnarray*}}
\newcommand{\inp}[2]{\langle {#1} ,\,{#2} \rangle}
\def\bmatrix#1{\left[ \begin{matrix} #1 \end{matrix} \right]}
\def\cmatrix#1{\left( \begin{matrix} #1 \end{matrix} \right)}
\def \noin{\noindent}
\newcommand{\evenindex}{\Pi_e}



\def \R{{\mathbb R}}
\def \C{{\mathbb C}}
\def \F{{\mathbb F}}
\def \K{{\mathbb K}}
\def \cu{\mathrm{CU}}
\def \calN{{\mathcal N}}
\def \T{{\mathbb T}}
\def \Pb{\mathrm{P}}
\def \N{{\mathbb N}}
\def \Ib{\mathrm{I}}
\def \Ls{{\Lambda}_{m-1}}
\def \Gb{\mathrm{G}}
\def \Hb{\mathrm{H}}
\def \Lam{{\Lambda}}

\def \Qb{\mathrm{Q}}
\def \Rb{\mathrm{R}}
\def \Mb{\mathrm{M}}
\def \norm{\nrm{\cdot}\equiv \nrm{\cdot}}

\def \A{{{\mathbb P}_1(\C^{n\times n})}}
\def \L{{\mathbb L}}
\def \G{{\F_{\tt{H}}}}
\def \S{\mathbb{S}}
\def \s{\mathbb{s}}
\def \sigmin{\sigma_{\min}}
\def \elam{\Lambda_{\epsilon}}
\def \slam{\Lambda^{\S}_{\epsilon}}
\def \Ib{\mathrm{I}}
\def \Tb{\mathrm{T}}
\def \d{{\delta}}

\def \Lb{\mathrm{L}}
\def \N{{\mathbb N}}
\def \Ls{{\Lambda}_{m-1}}
\def \Gb{\mathrm{G}}
\def \Hb{\mathrm{H}}
\def \Delta{\triangle}
\def \Rar{\Rightarrow}
\def \p{{\mathsf{p}(\lam; v)}}

\def \D{{\mathbb D}}

\def \tr{\mathrm{Tr}}
\def \cond{\mathrm{cond}}
\def \lam{\lambda}
\def \sig{\sigma}
\def \sign{\mathrm{sign}}

\def \ep{\epsilon}
\def \diag{\mathrm{diag}}
\def \rev{\mathrm{rev}}
\def \vec{\mathrm{vec}}

\def \ham{\mathsf{Ham}}
\def \herm{\mathsf{Herm}}
\def \sym{\mathsf{sym}}
\def \odd{\mathsf{sym}}
\def \en{\mathrm{even}}
\def \rank{\mathrm{rank}}
\def \pf{{\bf Proof: }}
\def \dist{\mathrm{dist}}
\def \rar{\rightarrow}

\def \rank{\mathrm{rank}}
\def \pf{{\bf Proof: }}
\def \dist{\mathrm{dist}}
\def \Re{\mathsf{Re}}
\def \Im{\mathsf{Im}}
\def \re{\mathsf{re}}
\def \im{\mathsf{im}}

\def \sym{\mathsf{CSym}}
\def \sksym{\mathsf{skew\mbox{-}sym}}
\def \odd{\mathrm{odd}}
\def \even{\mathrm{even}}
\def \herm{\mathsf{Herm}}
\def \skherm{\mathsf{skew\mbox{-}Herm}}
\def \str{\mathrm{ Struct}}
\def \cnot{\mathrm{CNOT}}
\def \eproof{$\blacksquare$}

\def \U{\mathsf{U}}
\def \G{\mathsf{G}}
\def \bS{{\bf S}}
\def \cA{{\cal A}}
\def \E{{\mathcal E}}
\def \X{{\mathcal X}}
\def \cH{\mathcal{H}}
\def \cJ{\mathcal{J}}
\def \tr{\mathrm{Tr}}
\def \range{\mathrm{Range}}
\def \adj{\star}

\def \pal{\mathrm{palindromic}}
\def \palpen{\mathrm{palindromic~~ pencil}}
\def \palpoly{\mathrm{palindromic~~ polynomial}}
\def \odd{\mathrm{odd}}
\def \even{\mathrm{even}}
\def \QT{{\texttt{QT}}}

\newcommand{\tm}[1]{\textcolor{magenta}{ #1}}
\newcommand{\tre}[1]{\textcolor{red}{ #1}}
\newcommand{\tb}[1]{\textcolor{blue}{ #1}}
\newcommand{\tg}[1]{\textcolor{green}{ #1}}

\title{Quantum embedding of graphs for subgraph counting}

\author{\IEEEauthorblockN{Bibhas Adhikari\thanks{Email: badhikari@fujitsu.com}}
\IEEEauthorblockA{\textit{Fujitsu Research of America Inc.} \\
Santa Clara, CA, USA }
}

\maketitle

\begin{abstract}
We develop a unified quantum framework for subgraph counting in graphs. We encode a graph on $N$ vertices into a quantum state on $2\lceil \log_2 N \rceil$ working qubits and $2$ ancilla qubits using its adjacency list, with worst-case gate complexity $O(N^2)$, which we refer to as the graph adjacency state. We design quantum measurement operators that capture the edge structure of a target subgraph, enabling estimation of its count via measurements on the $m$-fold tensor product of the adjacency state, where $m$ is the number of edges in the subgraph. We illustrate the framework for triangles, cycles, and cliques. This approach yields quantum logspace algorithms for motif counting, with no known classical counterpart.
\end{abstract}


\section{Introduction} 

Many real-world optimization problems can be formulated as combinatorial problems on graphs, where the graph encodes underlying data. To exploit quantum algorithms for such problems, classical data must be embedded into quantum states. For a survey of quantum data encoding techniques, including graph-structured data, see \cite{shayeganfar2025quantum}, and for the role of graphs in quantum computing, see \cite{romanello2025role}. Quantum graph embedding represents graph structures as quantum states, enabling the use of superposition and entanglement for computational tasks such as machine learning.  Existing approaches include variational quantum circuits, quantum annealing, and quantum-inspired methods.

Given a weighted (directed) graph $G^w=(V, E)$ on $N$ vertices with an weight function $w:V\times V\rightarrow \R_{\geq 0},$ we associate a quantum state defined as follows. Denoting the weight of an edge $(r,c)\in E$ as $w_{rc}$ we define a $2\ceil{\log_2N}$-qubit quantum state  
\begin{eqnarray}
    \ket{G^w} &=& \frac{1}{W} \sum_{r,c\in\{0,1\}^n} \bra{r}A(G^w)\ket{c} \,\, \ket{r}\ket{c} \nonumber\\ &=& \frac{1}{W} \sum_{r,c\in\{0,1\}^n} w_{rc} \,\, \ket{r}\ket{c},\label{eqn:adjws}
\end{eqnarray} where $W=\sqrt{\sum_{(r,c)\in E} w_{rc}^2}$ is the normalizing constant. Here, the adjacency matrix $A(G^w)$ is defined by 
\begin{equation}
    A(G^w)_{r,c} = \begin{cases}
        w_{rc}, \,\, \mbox{if} \,\, (r,c)\in E \\
        0, \,\, \mbox{otherwise.}
    \end{cases}
\end{equation} The notation $r$ and $c$ refer to the row and column indices of the adjacency matrix, respectively. In particular, $W=\sqrt{2\, |E|}$ for an unweighted simple undirected graph $G$ and $W=\sqrt{ |E|}$ for unweighted directed graph, where $|E|$ is the number of edges in the graph. 

We call the state given by equation (\ref{eqn:adjws}) as the \textit{graph adjacency state} since it encodes the adjacency relations of a graph. In this paper, we restrict our attention only to simple unweighted graphs. Essentially, we encode the adjacency matrix into a multi-qubit quantum state, where the number of qubits is logarithmic number of vertices. Thus it is different from block-encoding technique, which is used to encode matrices through a quantum circuit, for example see \cite{yosef2025encoding} and the references therein. A doubt is perpetuated in the literature about achieving advantages in combinatorial optimization by utilizing only $\log_2N$ qubits for $N$ data points. In this paper, we establish that multiple copies of such state can extend some potential avenues for gaining advantage in learning structural properties of a graph.

We propose quantum circuit constructions for preparing a graph adjacency state from the adjacency list of an unweighted  graph. The construction uses two intermediate states, the \emph{vertex neighborhood state} and the \emph{degree distribution state}, each implemented on $\log_2 N$ qubits with one ancilla, and achieves $O(N^2)$ worst-case gate complexity, however it depends on graph sparsity and vertex labeling. Building on this representation $\ket{G}$ for a graph $G$, we estimate subgraph frequencies by encoding the edge structure of a subgraph $S$ with $m$ edges into a projection operator $P_S$ with the aid of $m$ copies of $\ket{G},$ where $m$ is the number of edges in $S.$  

To validate the framework via numerical simulations, we consider Erd\H{o}s–R\'enyi graphs with varying edge probabilities and graph sizes $N$, considering triangles, 4-cycles, and 4-cliques. The estimator is implemented by sampling from Bernoulli trials with success probability $p=\bra{G}^{\otimes m}P_S\ket{G}^{\otimes m}$, which is computed exactly using combinatorial formulas (e.g., triangle, cycle, and clique counts), thereby faithfully reproducing the statistics of repeated quantum measurements without explicitly simulating the exponentially large state vector.

We show that the POVM-based method requires $m(2\lceil \log_2 N\rceil + 2)$ qubits per shot and achieves additive error $\epsilon$ with sample complexity $\Theta\!\left(\frac{1}{\epsilon^2}\log\frac{1}{\delta}\right)$. In contrast, amplitude amplification-based estimation reduces the query complexity to $O\!\left(\frac{1}{\epsilon}\log\frac{1}{\delta}\right)$ while maintaining the same asymptotic qubit requirement, at the cost of coherent circuits. For constant $m$, both approaches operate in quantum logarithmic space, and amplitude amplification provides a quadratic improvement in precision dependence. We emphasize that the proposed framework improves space complexities of existing classical and quantum algorithms (see Section \ref{Sec:somputecomplex}) and it is a unified framework for directed and undirected graphs. 

Note that, given $\epsilon, \delta, m$ the potential of the proposed framework lies in the fact that there is no classical logspace algorithms for motif counting in the literature. Recall that given a problem with input length $N,$ a logspace algorithm uses work space $O(\log N)$ and time $poly(N)$ which is equivalent to the number of local updates. The quantum logspace algorithms is the class of unitary quantum algorithms, which include storing information using qubits, whose local updates are unitary matrices, and the outcome of measuring some qubits at the end of the computation, see \cite{watrous1999space} \cite{fefferman2016complete}\cite{girish2024quantum} for further details.  

The remainder of the paper is organized as follows. Section \ref{sec:2} reviews some of the existing quantum state representations of a graph in the literature along with a review of existing quantum circuit proposals for preparing (sparse) quantum state. Section \ref{sec:3} includes the quantum circuit preparation of the adjacency graph state for undirected graph. Section \ref{sec:4} includes the measurement operators for $k$-clique, $k$-cycle counts and the sample complexity of the proposed methods along with numerical simulation results. 

\section{Preliminaries}\label{sec:2}
In this section, we briefly review some of the existing methods for quantum state representation of graphs and quantum circuit constructions for preparing sparse quantum states, along with a brief review of space computational complexities of existing classical and quantum methods for subgraph counting.

\subsection{Quantum state representation of graphs}

There are a few proposals for the representation of the quantum state of a given graph such as graph states \cite{briegel2001persistent} \cite{raussendorf2001one}, the quantum states defined by combinatorial Laplacian corresponding to simple graphs \cite{braunstein2006laplacian} \cite{de2016spectral} and complex-weighted directed graphs corresponding to Laplacian and signless Laplacians \cite{adhikari2017laplacian}, termed as the \textit{graph Laplacian states} \cite{dutta2019condition}. In another direction, a sequence of articles is written from a different perspective to analyze experimental data through graph structures, see \cite{gu2019quantum} and the references therein.  Different quantum machine learning approaches are proposed to deal with structured quantum data in graphs, such as associating vertices of a graph with density matrices, where the edges represent information-theoretic correlations of the neighboring vertices \cite{beer2023quantum}. Quantum analogues of classical neurons are developed to design quantum feedforward neural networks that are capable of universal quantum computation \cite{beer2020training}.

In \cite{ionicioiu2012encoding}, the authors propose an axiomatic approach to encode a graph into a quantum state by associating each vertex with a quantum state in some Hilbert space. A quantum state corresponding to a graph is defined in \cite{ai2024discrete} from the perspective of discrete-time quantum walks, here the probability amplitudes encodes the transition probabilities of the walker, see also \cite{della2024quantum}.  A framework, known as \textit{equivariant quantum graph circuits} for learning functions over graphs using parameterized
quantum circuits is proposed in \cite{mernyei2022equivariant}.

\subsection{Quantum circuits for preparing sparse quantum states}

There are several algorithms and quantum circuit designs proposed in literature for generation of a desired quantum state, particularly for sparse quantum states, see \cite{luo2024circuit} and the references therein. An $n$-qubit $d$-sparse quantum state refers to an $n$-qubit state with $d$ nonzero probability amplitudes, there are established protocols to produce such a given state as an output of a quantum circuit $U$ with or without the use of $m$ ancillaery qubits such that $$U\ket{0}^{\otimes n}\ket{0}^{\otimes m}=\sum_{l=0}^{d-1} \alpha_l\ket{i_l}\ket{0}^{\otimes m},$$ where $\ket{i_l}$ is a canonical basis element of $\C^{2^n}.$  The proposal in  \cite{gleinig2021efficient} gives a desired quantum circuit with $O(dn)$ CNOT gates and $O(d\log d+n)$ one-qubit gates without use of any ancillary gates.   Then a sequence of proposals are given to reduce the circuit complexity with or without the use of up to 2 ancillary qubits, see \cite{malvetti2021quantum} \cite{ramacciotti2024simple} \cite{mozafari2022efficient} \cite{park2019circuit} \cite{de2022double}. A significant improvement is made in \cite{mao2024towards}, in which the authors propose a sparse state generation algorithm that generates a circuit of size (total number of gates) $O\left(\frac{dn}{\log n}+n\right)$ with 2 ancillary qubits. In addition to that they also show that when $d\leq 2^{\delta n},$ $0\leq \delta <1$ then this circuit size is asymptotically optimal with the use of at most $poly(n)$ ancillary qubits. Focusing on circuit depth along with circuit size, quantum circuit with $O(dn\log n)$ ancillary qubits is proposed in \cite{zhang2022quantum} with depth complexity $O(\log dn).$ Besides, for any $m\geq 0,$ \cite{sun2023asymptotically} proposed an algorithm with depth complexity $O\left(n\log dn + \frac{dn^2\log d}{n+m}\right)$ using $m$ ancillary qubits. We also refer to \cite{luo2024circuit} for recent development for quantum circuit complexity of preparing sparse quantum states. 

\subsection{Subgraph counting}\label{Sec:somputecomplex}

Subgraph counting is a fundamental task in various applications in graph structured data analysis \cite{ribeiro2021survey}. Indeed, this is one of those hard computational problems. In particular, determining if one subgraph exists at all in another larger network is an NP-Complete problem \cite{cook2023complexity}. There are several classical algorithms proposed in the literature to estimate the subgraph frequencies, in particular for subgraphs that are triangles, cycles, cliques or \textit{motif}s. Graph motifs are small subgraph patterns that appear in a network significantly more often than expected in suitably randomized null models, and they are widely regarded as the fundamental building blocks of complex networks. Figure \ref{fig:motifs} exhibits a few examples of motifs in directed and undirected graphs. Introduced systematically by Milo et al., motifs capture mesoscopic structural information that is invisible to global statistics such as degree distributions or clustering coefficients \cite{Milo2002}. Different domains exhibit characteristic motif profiles: feed-forward loops in transcriptional regulatory networks, triangles in social networks, and specific cycle structures in technological and power-grid networks \cite{Alon2007} \cite{Newman2010}. Motif analysis has been used to infer functional organization, robustness, and information-processing capabilities of networks, and has found applications in biology, neuroscience, finance, and infrastructure systems \cite{Sporns2004} \cite{Battiston2016}. From a computational perspective, counting motifs such as triangles, $k$-cycles, and $k$-cliques poses significant algorithmic challenges and has motivated the development of scalable, exact, approximate, and streaming algorithms \cite{Latapy2008} \cite{Bressan2019}.

\begin{figure}[htbp]
     \centering
     \includegraphics[width=0.5\textwidth]{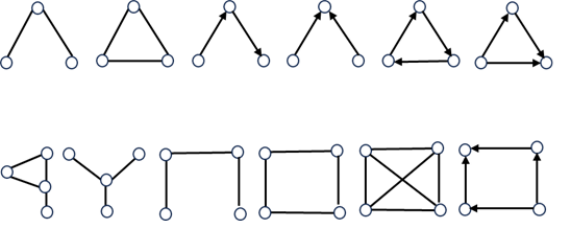}
\centering \caption {A few examples of motifs in (directed) graphs.}
    \label{fig:motifs}
\end{figure} 

The space complexity of these algorithms depend on several parameters associated with a graph including the number of edges and the number of vertices, see the recent survey \cite{yin2025motif}. To be specific, most of the in memory $k$-clique algorithms the space complexity is $O(|E|+N).$ Besides, there are several excellent approaches including approximation algorithms, heterogeneous and parallel algorithms, sample-based algorithms etc.  A naive brute-force approach for $k$-clique detection
has a worst-case computational complexity of $O(N^kk^2),$ which is reduced by using optimizers up to $O(N^c)$ in space, where $c$ is independent of $k$ in \cite{vassilevska2009efficient}. On the other hand, a few quantum approaches are proposed to estimate $k$-clique counts in an undireced graph, especially for triangles. In \cite{perriello2025quantum}, the authors nicely describe the quantum algorithms in literature for $k$-clique finding in undirected graphs such as \cite{metwalli2020finding} which requires $O\left(N\log_2k^2\right)$ qubits, \cite{haverly2021implementation} which requires $O(N^2)$ qubits, the best known result achieves query complexity $O(N^{5/4})$ for triangle counts in \cite{le2014improved}. Recently, \cite{perriello2025quantum} proposes a Quantum Amplitude Amplification-based and a quantum walk-based methods using a state which is closely related to the adjacency quantum state for $k$-clique counting in an undirected graph with $O(N\log_2N)$ qubits. 

\section{Quantum circuits for preparing graph adjacency state}\label{sec:3}

In this section, we develop quantum circuit models to prepare the $2n$-qubit quantum state \begin{equation}
    \ket{G}=\frac{1}{\sqrt{2|E|}} \sum_{r,c} a_{rc} \ket{r}\ket{c}
\end{equation} corresponding to a graph $G$ on $N=2^n$ vertices, $|E|$ denotes the number of edges of $G,$ and $a_{rc}$ denotes the $(r,s)$-entry of the adjacency matrix associated with $G.$

First note that \begin{eqnarray*}
   \ket{G} &=& \frac{1}{\sqrt{2|E|}}\sum_{\substack{r,c=0\\ r\neq c}}^{N-1} a_{rc} \ket{r} \ket{c}   \\
   &=& \sum_{r=0}^{N-1} \sqrt{\frac{d_r}{2|E|}}\ket{r}\left(\sum_{c\in \mathcal{N}_r} \frac{1}{\sqrt{d_r}} a_{rc}\ket{c}\right), \nonumber
\end{eqnarray*} where $\mathcal{N}_r$ denotes the set of neighbors of $c$ i.e. $\mathcal{N}_r=\{c\in V : a_{rc}\neq 0\}.$ To prepare the state $\ket{G},$ it is enough to implement two unitary transformations that provide \begin{eqnarray}
\mathsf{U}_D &:& \, \ket{0}^{\otimes n}\mapsto \sum_{r=0}^{N-1} \sqrt{\frac{d_r}{2|E|}}\ket{r} \,\, \mbox{and}  \nonumber \\ \U^{(r)} &:& \, \ket{0}^{\otimes n}\mapsto \ket{\calN_r}=\sum_{c\in \mathcal{N}_r} \frac{1}{\sqrt{d_r}} \ket{c},\label{eqn:nstate}\end{eqnarray} where $d_r=|\calN_r|$ is the degree of the vertex $r.$ The fundamental difference between these two output states is that the later has equal probability amplitudes, whereas the former has the same property if and only if the graph is regular. The former state encodes the degrees of the vertices on a graph with logarithmic overhead, and hence we call it the \textit{degree distribution state} for a given graph.  Besides, the later quantum state defines the state which encodes the neighbors of a vertex $r,$ and hence we call it the \textit{vertex neighborhood quantum state} for a given vertex.  

We propose quantum circuit models for implementing $\U_D$ and $\U^{(r)}$, which leads to the construction of the graph adjacency state $\ket{G}$ using two quantum registers $\mathsf{QR}_1$ and $\mathsf{QR}_2$ respectively, each is composed of $\lceil\log_2N\rceil +1$ qubits, in which one qubit act as an ancilla qubit. The implementation of these circuits are given in Algorithm \ref{algo:U_sp} and \ref{algo:U_r} for $\U_D$ and $\U^{(r)},$ respectively. Now having these circuit models, we discuss the following procedure for preparing $\ket{G}.$  The quantum circuit presented in Figure \ref{fig:qc_G} implement the transformation $\ket{0}^{\otimes\lceil\log_2N\rceil}\ket{0}^{\otimes\lceil\log_2N\rceil}$ $\mapsto$ $\ket{G}$ with the help of $2$ ancillary qubits with intial state $\ket{0}.$ The ancillary qubits are represented with red wires, whereas the bold black colored wires denote the working qubit registers. Although the circuit is self-explanatory, we describe the implementing the circuit in Algorithm \ref{algo:U_DU_r}. 

\begin{figure}[htbp]
     \centering
     \includegraphics[width=1.0\linewidth]{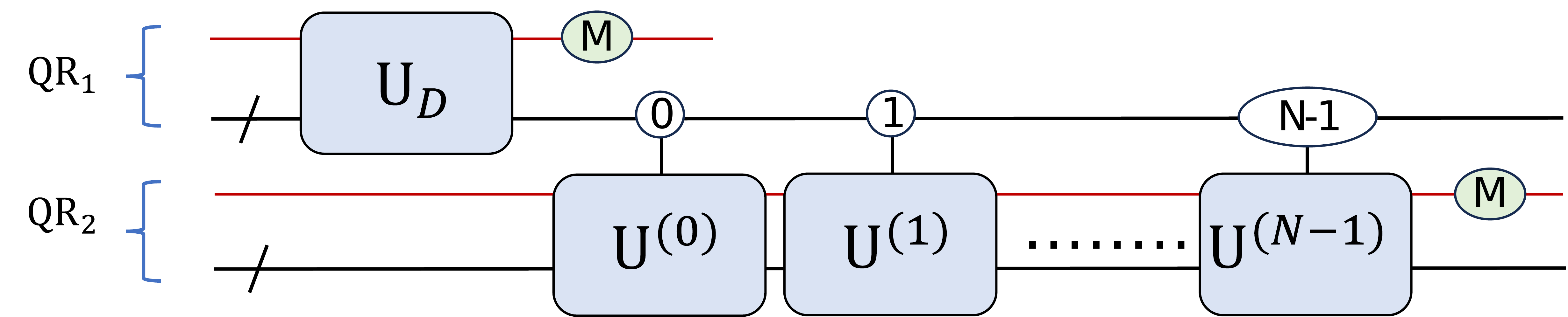}
\centering \caption { The quantum circuit for preparing $\ket{G}$.}
    \label{fig:qc_G}
\end{figure}

\begin{algorithm}
\caption{Preparing the adjacency state $\ket{G}$}\label{algo:U_DU_r}
\textbf{Input:} $\ket{0}^{\otimes\lceil\log_2N\rceil}$ for working qubits in register $\mathsf{QR}_1,$ $\mathsf{QR}_2$; $\ket{0}$ for ancillary qubits, the  quantum circuits for $\U_D$ and $\U^{(r)}$ for each vertex $r$ (Algorithms \ref{algo:U_sp} and \ref{algo:U_r}).\\
\textbf{Output:} $\ket{G}$ on the working qubits register
\begin{algorithmic}
\State Implement $\U_D$ by measuring the ancilla qubit wrt Pauli $Z$ operator in $\mathsf{QR}_1$ upon receiving $\ket{1}$ 
\For{$r=0,1,\hdots,N-1$ }
\State Implement the gate $\U^{(r)}$ in $\mathsf{QR}_2$ controlled by working qubits in $\mathsf{QR}_1$ with state $\ket{r}$ through the binary representation of $r$ 
\State Perform a Pauli $Z$ measurement on the ancilla qubit in $\mathsf{QR}_2$
\EndFor
\end{algorithmic}
\end{algorithm}

\begin{theorem}
 The quantum circuit in Algorithm \ref{algo:U_DU_r} implements a unitary transformation $\ket{0}^{\otimes 2\lceil\log_2N\rceil}$ $\mapsto$ $\ket{G}$ in the working qubit register formed by $\mathsf{QR}_1$ and $\mathsf{QR}_2$ with the aid of $2$ ancillary qubits with probability one. The circuit can be implemented using $N+2|E|$ number of $(\lceil\log_2N\rceil+1)$-qubit Toffoli gates, and $2\sum_{r=1}^{N-1} H(0,r) + 2\sum_{r=0}^{N-1}\sum_{c\in\calN_r} H(0,c) -4|E|$ CNOT gates, where $H(x,y)$ denotes the Hamming distance between binary representation of two integers $x,y\in\{0,\hdots,N-1\}$.   
\end{theorem}

\pf The entire implementation of the circuit is based on the quantum registers $\mathsf{QR}_1$ and $\mathsf{QR}_2$ incorporating the gates corresponding to $\U_D$, and $\U^{(r)}$ for each $r\in\{0,1,\hdots,N-1\}$. First, performing the measurement on the ancilla qubit in $\mathsf{QR}_1$ implements $\U_D$ upon receiving the measurement outcome $\ket{1}$, which implements 
\begin{eqnarray*}
&& \ket{0}_1\ket{0}_2\ket{0}^{\lceil\log_2N\rceil}\ket{0}^{\lceil\log_2N\rceil} \mapsto \\ && \ket{1}_1\ket{0}_2\left(\sum_{r=0}^{N-1} \sqrt{\frac{d_r}{2|E|}}\ket{r}\ket{0}^{\lceil\log_2N\rceil}\right),
\end{eqnarray*} where $\ket{a}_j$ denotes the state of the ancilla qubit in $\mathsf{QR}_j,$ $a_j\in\{0,1\},$ $j=1,2.$ Then applying the sequence of controlled $\U^{(r)}$ gates on $\mathsf{QR}_2$ the desired state $\ket{G}$ is obtained in the working qubits through the transformation 
\begin{eqnarray*}
&& \ket{1}_1\ket{0}_2\left(\sum_{r=0}^{N-1} \sqrt{\frac{d_r}{2|E|}}\ket{r}\right)\ket{0}^{\lceil\log_2N\rceil}\mapsto \\ && \ket{1}_1\ket{1}_2\sum_{r=0}^{N-1} \sqrt{\frac{d_r}{2|E|}}\ket{r}\left(\sum_{c\in\calN_r} \frac{1}{\sqrt{d_r}}a_{rc}\ket{c}\right).    
\end{eqnarray*} The Pauli $Z$ measurements on both the ancilla qubits prepares $\ket{G}$ with probability one.

The gate counts follow from Theorem \ref{thm:UD} and Theorem \ref{thm:Ur}. This completes the proof. \hfill{$\square$}

We denote a controlled single-qubit unitary gate $U$ as $CU_{\{x,y\}}$ corresponding to the control qubit $x$ and target qubit $y.$ We also denote $T^{(n+1)}_{\{c,x\}}$ as the Toffoli gate on an $(n+1)$-qubit register with the control $n$ qubit state is given by $\ket{c}=\ket{c_{n-1}\cdots c_1c_0}$ corresponding to the $n$-bit representation of a positive integer $c=(c_{n-1},\hdots,c_1,c_0)\in\{0,1\}^n,$ where as $x$ denotes the target qubit.

\subsection{Preparing vertex neighborhood quantum state}

Given a vertex $r\in V$ of a graph on $N$ vertices and its adjacency list or the set of neighbors of $\calN_r\subset V,$ we design a quantum circuit on the $\lceil\log_2 N\rceil$-qubit register for preparing the quantum state $\ket{\calN_r}$ as defined in equation (\ref{eqn:nstate}). For brevity we assume $N=2^n$.

Let $\calN_r=\{c_0,\hdots,c_{d_r-1}\}$ with $c_{d_r-1}\leq \cdots\leq c_1\leq c_0.$ Suppose $c_j=(c_{n-1}^{(j)},\cdots, c_{1}^{(j)}, c_{0}^{(j)})$ is the $n$-bit representation of $c_j,$ $0\leq j\leq n-1$ such that $c_j=\sum_{l=0}^{n-1} c_{l}^{(j)} 2^l.$ For $d_r\geq 1,$ we define a collection of unitary gates $U_{\sqrt{d_r-j}}\in \mathfrak{U}(2)$ given by \begin{equation}
    U_{\sqrt{d_r-j}}=\bmatrix{\sqrt{\frac{d_r-j-1}{d_r-j}} & - \sqrt{\frac{1}{d_r-j}}  \\  \sqrt{\frac{1}{d_r-j}}  & \sqrt{\frac{d_r-j-1}{d_r-j}}},
\end{equation} for $j\in\{0,1,\hdots,d_r-1\}.$ In particular, note that $CU_{\sqrt{d_r-j}}=\bmatrix{\frac{1}{\sqrt{2}} & -\frac{1}{\sqrt{2}} \\ \frac{1}{\sqrt{2}} & \frac{1}{\sqrt{2}}}=H,$ the Hadamard gate for $j=d_r-2$ and $CU_{\sqrt{d_r-j}}=\bmatrix{0&-1\\1&0}=-\iota Y$ for $j=d_r-1.$ Then we provide Algorithm \ref{algo:U_r} for the construction of the circuit representation of the unitary operator $\U^{(r)}$ on a $(\log_2N+1)$-qubit register $\mathsf{QR}_1$ with one ancillary qubit $\ket{a},$  the first qubit in the register, whose measurement with respect to Pauli $Z$ operator implements $\ket{0}^{\otimes n}\mapsto \ket{\calN_r}$ upon receiving the outcome $\ket{1}$ for any vertex $r$ of a given graph $G.$ 

\begin{algorithm}
\caption{Circuit implementation of $\U^{(r)}$}\label{algo:U_r}
\textbf{Input:} $\mathcal{N}_r=\{c_0,c_1,\hdots,c_{d_r-1}\}$ such that $c_0\geq c_1\geq \hdots \geq c_{d_r-1}$ with $c_j=(c_{n-1}^{(j)},\hdots,c_1^{(j)},c_0^{(j)})$ as the $n$-bit  representation of $c_j.$ The ancilla qubit state is denoted as $\ket{a},$ $a\in\{0,1\}.$ The labeling of the $n$ qubits in the quantum register $\mathsf{QR}_1$ is the ordered sequence $n-1,\hdots,1,0$ from top to bottom.\\
\textbf{Output:} $\sum_{c\in\mathcal{N}_r} \sqrt{\frac{1}{d_r}} a_{rc}\ket{c}$ 
\begin{algorithmic}
\For{$j=0,1,\hdots,d_r-1$ }
\For{$j=0,1,\hdots,d_r-2$ }
\If{$k\in\{0,1,\hdots,n-1\}$ is the least integer such that $c_k^{(j)}=1$}
\State implement the control gate ${CU_{\sqrt{d_r-j}}}_{\{a=0,k\}}$
\State implement $CNOT_{\{a=0,k\neq l:c_{l}^{(j)}=1\}}$
\State implement $T_{\{c_j,a\}}^{(n+1)}$ 
\State implement $CNOT_{\{a=0,k\neq l:c_{l}^{(j)}=1\}}$
\EndIf
\EndFor
\If{$c_{d_r-1}\neq 0$}
\State do as above
\EndIf
\If{$c_{d_r-1}=0$}
\State implement the $T_{\{c_{d_r-1},a\}}^{(n)}$ 
\EndIf
\EndFor
\textbf{Measurement:} Measure the ancilla-qubit wrt Pauli $Z$ 
\end{algorithmic}
\end{algorithm}

\begin{figure}[htbp]
     \centering
     \subfigure [\centering ]
{\includegraphics[width=0.3\linewidth]{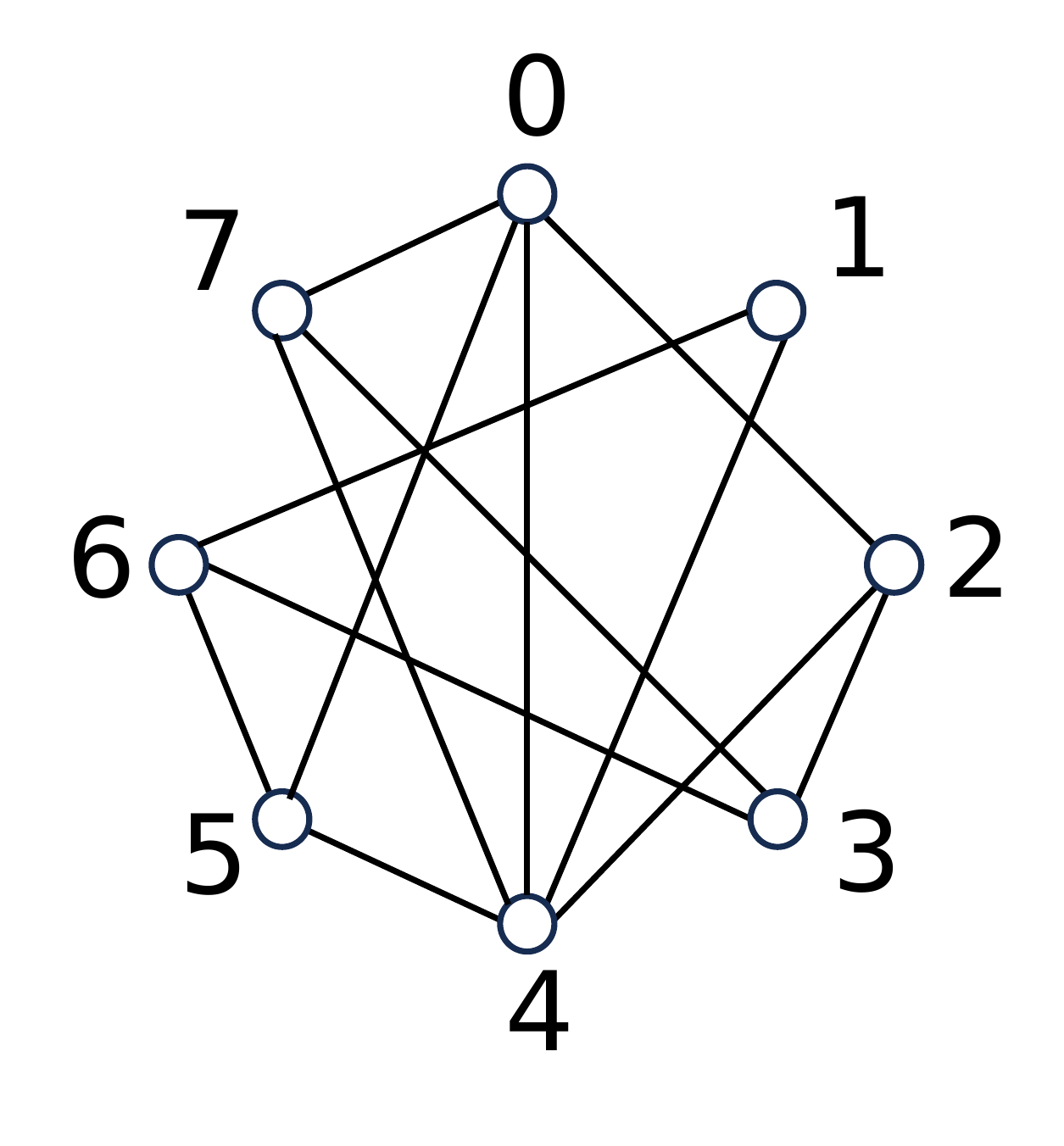}}
 \qquad
     \subfigure [\centering ]
{\includegraphics[width=1.0\linewidth]{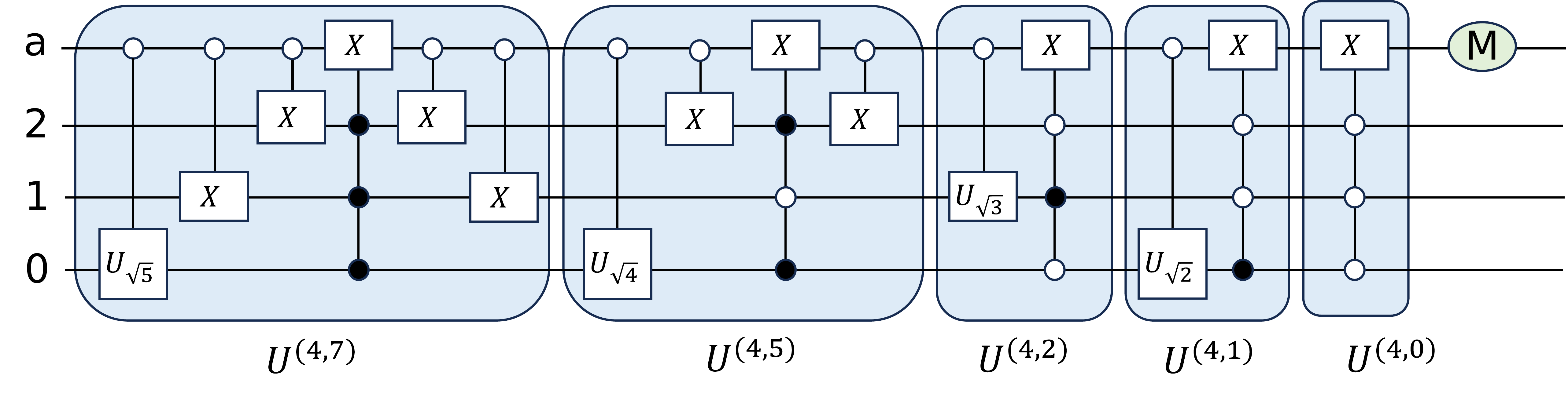}}
\centering \caption {(a) A random graph on $8$ vertices. (b) The quantum circuit following Algorithm \ref{algo:U_r} for $U^{(4)}.$}
    \label{fig:exp_G8}
\end{figure}

Note that the quantum circuit representation of $\U^{(r)}$ given by Algorithm \ref{algo:U_r} is expressed as $$\U^{(r)}=\prod_{j=d_r-1}^0 U^{(r,c_j)}: \ket{0}\ket{0}^{\otimes n}\mapsto \sum_{j=0}^{d_r-1} \frac{1}{\sqrt{d_r}}\ket{1}\ket{c_j},$$ where $U^{(r,c_j)}$ implements the map $\ket{0}\ket{0}^{\otimes n}\mapsto \sqrt{\frac{d_r-j-1}{d_r-j}}\ket{0}\ket{0}^{\otimes n} + \sqrt{\frac{1}{d_r-j}}\ket{1}\ket{c_j}$ for $0\leq j\leq d_r-1.$     

For example, consider the randomly generated graph $G$ in Figure \ref{fig:exp_G8} (a) on eight vertices. The circuit in Figure \ref{fig:exp_G8} (b) implements the operator $\U^{(4)}$ for generation of $\ket{\calN_4}.$ Note that the highest degree vertex of $G$ is $r=4$ with $\calN_4=\{c_0=7, c_1=5, c_2=2, c_3=1, c_4=0 (c_{d_4-1})\}$ with the $3$-bit representation $0=(0,0,0),$ $1=(0,0,1),$ $2=(0,1,0),$ $4=(1,0,0)$ $5=(1,0,1)$, and $7=(1,1,1).$ Then the unitary gates $CU_{\sqrt{d_r-j}}$ for $j=0,1,\hdots,4$ are given by 
\begin{eqnarray*}
    && CU_{\sqrt{5}}=\bmatrix{\sqrt{\frac{4}{5}} & - \sqrt{\frac{1}{5}} \\ \sqrt{\frac{1}{5}} & \sqrt{\frac{4}{5}}}, CU_{\sqrt{4}}=\bmatrix{\sqrt{\frac{3}{4}} & - \sqrt{\frac{1}{4}} \\ \sqrt{\frac{1}{4}} & \sqrt{\frac{3}{4}}}, \\ && CU_{\sqrt{3}}=\bmatrix{\sqrt{\frac{2}{3}} & - \sqrt{\frac{1}{3}} \\ \sqrt{\frac{1}{3}} & \sqrt{\frac{2}{3}}}, CU_{\sqrt{2}}=H, CU_{\sqrt{1}}=-\iota Y,
\end{eqnarray*} respectively.

Then \begin{eqnarray*}
 && \U^{(4)}\ket{0}\ket{000}  =\prod_{j\in\mathcal{N}_4} U^{(4,j)}\ket{0}\ket{000} \\ &=& U^{(4,0)} U^{(4,1)} U^{(4,2)} U^{(4,5)} U^{(4,7)}  \ket{0}\ket{000} \\
   &=& \sqrt{\frac{1}{5}}\ket{1}\ket{000}+\sqrt{\frac{1}{5}}\ket{1}\ket{001}+\sqrt{\frac{1}{5}}\ket{1}\ket{010}\\&& +\sqrt{\frac{1}{5}}\ket{1}\ket{101}+\sqrt{\frac{1}{5}}\ket{1}\ket{111}.
\end{eqnarray*} Then, obviously, measuring the observable Pauli $Z$ operator on ancilla qubit, the desired state is obtained with probability $1$.  

\begin{theorem}\label{thm:Ur}
Given a vertex $r$ in a graph $G=(V,E),$ the quantum circuit implementation of $\U^{(r)}$ is composed of $d_r$ two-qubit control unitary gates, $d_r$ number of $(\ceil{\log_2N}+1)$-qubit Toffoli gates, and $2\sum_{c\in \mathcal{N}_r} H(0,c)-2d_r$ CNOT gates, where $H(0,c)$ denotes the Hamming distance between the $\ceil{\log_2N}$-bit representation of $0,c\in V.$    
\end{theorem}
\pf Let $\mathcal{N}_r=\{c_0,\hdots,c_{d_r}\}.$ For each $c_j\in\mathcal{N}_r,$ the two-qubit unitary gate $CU_{\sqrt{d_r-j}}$, is used to implement the unitary map $\ket{0}\ket{0}^{\otimes n}\mapsto \sqrt{\frac{d_r-j-1}{d_r-j}}\ket{0}\ket{0}^{\otimes n}+ \sqrt{\frac{1}{d_r-j}}\ket{0}\ket{0}^{\otimes (k-1)}\ket{1}\ket{0}^{\otimes (n-k)}$, where $k\in\{0,1,\hdots,n-1\}$ is the least positive integer such that $c^{(j)}_k=1$ in the $n$-bit representation of the vertex label $c_j=(c^{(j)}_{n-1},\hdots,c_1^{(j)},c_0^{(j)}),$ for each $j.$ Thus a total of $d_r$ two-qubit control unitary gates $CU_{\sqrt{d_r-j}}$ is required if $0\notin \calN_r$, otherwise it requires $d_r$ such gates. Then for each $l\neq k$ such that $c^{(j)}_l=1,$ the CNOT gate with control ancillary qubit with state $\ket{0}$, the target qubit correspond to the $l$-th qubit in the $n$-qubit quantum register. This implements the map $\sqrt{\frac{d_r-j-1}{d_r-j}}\ket{0}\ket{0}^{\otimes n}+ \sqrt{\frac{1}{d_r-j}}\ket{0}\ket{0}^{\otimes (k-1)}\ket{1}\ket{0}^{\otimes (n-k)} \mapsto \sqrt{\frac{d_r-j-1}{d_r-j}}\ket{0}\ket{c_j-2^{k}}+\sqrt{\frac{1}{d_r-j}}\ket{0}\ket{c_j}.$ Since there are $H(0,c_j)-1$ number of such $l$ such that $c^{(j)}_l=1,$ $(l\neq k)$ hence the number of CNOT gates till uesed in the circuit is $H(0,c_j)-1.$ Further, for $c_j,$ the $(n+1)$-qubit Toffoli gate $T^{(n)}_{c_j,a}$ with $n$-qubit control state is $\ket{c_j}$ and the $X$ gate is applied on the ancillary qubit implements the map: $\sqrt{\frac{d_r-j-1}{d_r-j}}\ket{0}\ket{c_j-2^{k}}+ \sqrt{\frac{1}{d_r-j}}\ket{0}\ket{c_j}\mapsto \sqrt{\frac{d_r-j-1}{d_r-j}}\ket{0}\ket{c_j-2^{k}}+\sqrt{\frac{1}{d_r-j}}\ket{1}\ket{c_j}.$ Further, the same set of CNOT gates, as used previously are used to implement the map: $\sqrt{\frac{d_r-j-1}{d_r-j}}\ket{0}\ket{c_j-2^{k}}+ \sqrt{\frac{1}{d_r-j}}\ket{1}\ket{c_j} \mapsto \sqrt{\frac{d_r-j-1}{d_r-j}}\ket{0}\ket{0}^{\otimes n}+ \sqrt{\frac{1}{d_r-j}}\ket{1}\ket{c_j}.$ This completes the proof. \hfill{$\square$}


\subsection{Preparing degree distribution state}

In this section, we propose a quantum circuit for preparing an arbitrary $n$-qubit state \begin{equation}\ket{\psi}=\sum_{j=0}^{N-1} \alpha_j\ket{j}, \, \alpha_j\in\C, j\in\{1,\hdots,N-1\},\end{equation} $\alpha_0\geq 0, \sum_{j=0}^{N-1} |\alpha_j|^2=1, \,  N=2^n$ using an ancilla qubit from an input state $\ket{0}^{\otimes n}\ket{0}$. Then we employ this circuit for preparing the degree distribution state corresponding to a graph as defined in equation (\ref{eqn:nstate}) as a special case. The vector formed by the coefficients of the state $\ket{\psi}$ is denoted as $\boldsymbol{\alpha}=(\alpha_0,\alpha_1,\hdots,\alpha_{N-1}).$

Given $\ket{\psi}$, we define the single-qubit gates

$$G_{\alpha_j} =
\bmatrix{
\sqrt{1-\frac{|\alpha_j|^2}{1-\sum_{l=1}^{j-1}|\alpha_l|^2}} & -\frac{\overline{\alpha_j}}{\sqrt{1-\sum_{l=1}^{j-1}|\alpha_l|^2}} \\[6pt]
\frac{\alpha_j}{\sqrt{1-\sum_{l=1}^{j-1}|\alpha_l|^2}} & \sqrt{1-\frac{|\alpha_j|^2}{1-\sum_{l=1}^{j-1}|\alpha_l|^2}}},$$ $j \geq 2,$
and
$$G_{\alpha_1} =
\bmatrix{
\sqrt{1-|\alpha_1|^2} & -\overline{\alpha_1} \\
\alpha_1 & \sqrt{1-|\alpha_1|^2}},$$
where $\overline{\alpha}$ denotes the complex conjugate of $\alpha \in \mathbb{C}$. It is straightforward to verify that both $G_{\alpha_j}$ and $G_{\alpha_1}$ are Hermitian and unitary matrices.

We now propose Algorithm~\ref{algo:U_sp}, which implements a unitary operator $\mathcal{U}_{sp}(\boldsymbol{\alpha})$ that prepares the state $\ket{\psi}$, such that a measurement of the ancilla qubit in the Pauli-$Z$ basis yields the desired outcome with probability $1$.

\begin{algorithm}
\caption{Circuit implementation of $\U_{sp}(\boldsymbol{\alpha})$ for general $n$-qubit states}\label{algo:U_sp}
\textbf{Description:} $\boldsymbol{\alpha}=\left(\alpha_0,\alpha_1,\hdots,\alpha_{N-1}\right)$ such that $\alpha_0\geq 0,$ $\sum_{j=0}^{N-1} |\alpha_j|^2=1.$ Let $(j_{n-1},\hdots,j_1,j_0)$ be the binary representation of $j$ such that $j=\sum_{l=0}^{n-1} j_l2^l$ Then ancilla qubit state is denoted as $\ket{a},$ $a\in\{0,1\}.$ \\ 
\textbf{Input:} $\ket{0}^{\otimes n}\ket{0}$\\ 
\textbf{Output:} $\sum_{j=0}^{N-1}  \alpha_{j}\ket{j}$ 
\begin{algorithmic}
\If{$\alpha_j\neq 0$} 
\For{$j=1,\hdots,N-1$ }
\If{$k$ is the least integer such that $j_k=1$}
\State implement the control gate $G_{\alpha_j}$ with target qubit $k$, control-qubit $a=0$ 
\EndIf
\If{$j_l=1$, $l\neq k$} 
\State implement $CNOT_{a=0,l}$ 
\EndIf
\State implement $T_{j,a}^{(n+1)}$ 
\If{$j_l=1$, $l\neq k$}
\State implement $CNOT_{a=0,l}$ 
\EndIf
\EndFor
\For{$j=0$ }
\State implement $T_{0,a}^{(n+1)}$ 
\EndFor
\EndIf
\textbf{Measurement:} Measure the ancilla-qubit wrt Pauli $Z$ 
\end{algorithmic}
\end{algorithm}

From Algorithm \ref{algo:U_sp} it is clear that for each vertex $j$ in $G$ there is a unitary gate $U^{(j)}_G$ composed of a controlled $G_{\alpha_j}$ and a number of CNOT gates decided by the Hamming weight of $j.$ Thus we have $$\U_{sp}(\boldsymbol{\alpha})=U_{\alpha_{(N-1),\hdots,\alpha_1,\alpha_0}} :=\prod_{j=0}^{N-1} U^{(j)},$$ where $U^{(j)},$ $j\neq 0$ implements the following map: \begin{eqnarray*}\ket{0}\ket{0}^{\otimes n} &\mapsto& \sqrt{1-\dfrac{|\alpha_j|^2}{1-\sum_{l=1}^{j-1}|\alpha_l|^2}} \ket{0} \ket{0}^{\otimes n} \\ && + \dfrac{\alpha_j}{\sqrt{1-\sum_{l=1}^{j-1}|\alpha_l|^2}} \ket{1}\ket{j},\end{eqnarray*} and $U^{(0)}$ maps $\ket{0}\ket{0}^{\otimes n}$ to $\ket{1}\ket{0}^{\otimes n}.$

\begin{theorem}\label{thm:UD}
The quantum circuit implementation of $\U_{sp}(\boldsymbol{\alpha}),$ $\boldsymbol{\alpha}=\left(\alpha_0,\alpha_1,\hdots,\alpha_{N-1}\right)$ is composed of $d$ two-qubit control unitary gates, $d$ number of $(\lceil\log_2N\rceil+1)$-qubit Toffoli gates, and $2\sum_{j=1:\alpha_j\neq 0}^{N-1} H(0,j)$ CNOT gates, where $H(0,j)$ denotes the Hamming distance between the $\ceil{\log_2N}$-bit representations of $0$ and $j.$ Here $d$ is the number of indices $j$ such that $\alpha_j\neq 0.$     
\end{theorem}

\pf  The proof is similar to the proof of Theorem \ref{thm:Ur}.

\begin{figure}[htbp]
     \centering
     \includegraphics[width=1.0\linewidth]{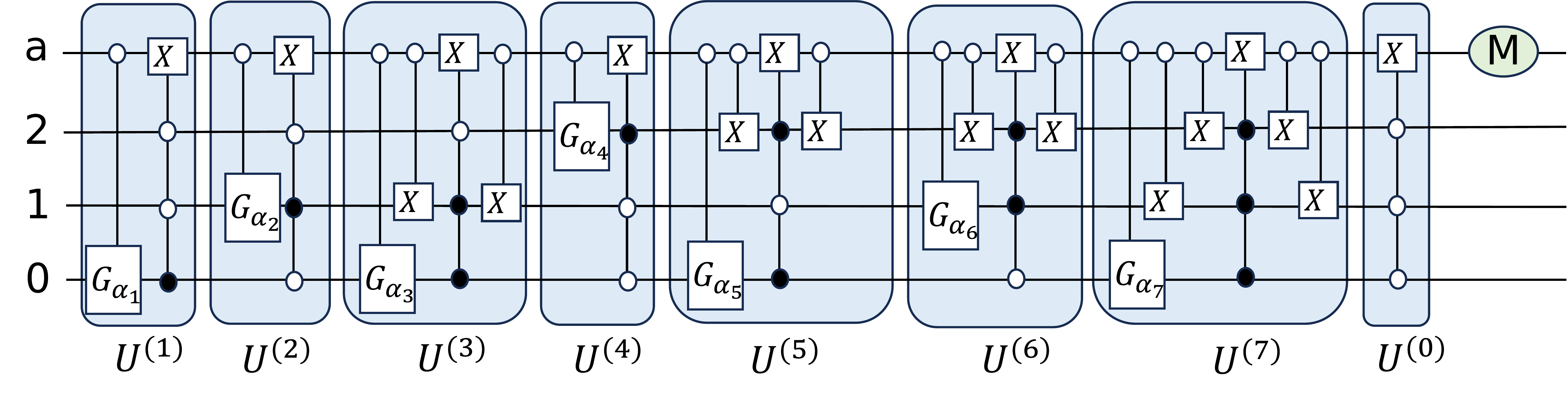}
\centering \caption { The quantum circuit following Algorithm \ref{algo:U_sp} for preparing degree distribution state for the graph in Figure \ref{fig:exp_G8} (a).}
    \label{fig:qc_exp_G8}
\end{figure}

Consequently, setting $\alpha_j=\sqrt{\frac{d_j}{2E}}$, the operator $\U_{sp}(\boldsymbol{\alpha})$ represents $\U_D$ and the  degree distribution state for a graph can be obtained. Here  $E$ denotes the number of edges and $d_j$ denotes the degree of vertex $j$.  For example, in Figure \ref{fig:qc_exp_G8}, we present the quantum circuit for the graph in Figure \ref{fig:exp_G8} (a). Note that $d_0=4,$ $d_1=2,$ $d_2=3,$ $d_3=3,$ $d_4=5,$ $d_5=3,$ $d_6=3,$ $d_7=3$, and hence $2E=26.$ Then the degree distribution state is given by \begin{eqnarray*}
    \ket{G_d} &=& \sqrt{\frac{4}{26}} \ket{0} + \sqrt{\frac{2}{26}} \ket{1} + \sqrt{\frac{3}{26}} \ket{2}+ \sqrt{\frac{3}{26}} \ket{3} \\ && + \sqrt{\frac{5}{26}} \ket{4} + \sqrt{\frac{3}{26}} \ket{5} + \sqrt{\frac{3}{26}} \ket{6} + \sqrt{\frac{3}{26}} \ket{7}.
\end{eqnarray*} Then the Algorithm \ref{algo:U_sp} gives the circuit in Figure \ref{fig:qc_exp_G8} to generate $\ket{G}_d$ with the input state $\ket{0}\ket{000}.$ Here
\begin{eqnarray*}
   && G_{\alpha_1}=\bmatrix{\sqrt{\frac{12}{13}} & -\sqrt{\frac{1}{13}} \\ \sqrt{\frac{1}{13}}& \sqrt{\frac{12}{13}} }, \, G_{\alpha_2}=\bmatrix{\sqrt{\frac{7}{8}} & -\sqrt{\frac{1}{8}} \\ \sqrt{\frac{1}{8}}& \sqrt{\frac{7}{8}} }, \\ && \,G_{\alpha_3}=\bmatrix{\sqrt{\frac{6}{7}} & -\sqrt{\frac{1}{7}} \\ \sqrt{\frac{1}{7}}& \sqrt{\frac{6}{7}} }, \, G_{\alpha_4}=\bmatrix{\sqrt{\frac{13}{18}} & -\sqrt{\frac{5}{18}} \\ \sqrt{\frac{5}{18}}& \sqrt{\frac{13}{18}} } \\
   && G_{\alpha_5}=\bmatrix{\sqrt{\frac{10}{13}} & -\sqrt{\frac{3}{13}} \\ \sqrt{\frac{3}{13}}& \sqrt{\frac{10}{13}} }, \, G_{\alpha_6}=\bmatrix{\sqrt{\frac{7}{10}} & -\sqrt{\frac{3}{10}} \\ \sqrt{\frac{3}{10}}& \sqrt{\frac{7}{13}} }, \\ && G_{\alpha_7}=\bmatrix{\sqrt{\frac{4}{7}} & -\sqrt{\frac{3}{7}} \\ \sqrt{\frac{3}{7}}& \sqrt{\frac{4}{7}} }.
\end{eqnarray*} Then $\U_{D}\ket{0}\ket{000} = U^{(0)}U^{(7)}\hdots U^{(1)} \ket{0}\ket{000}=\ket{1}\ket{G_d}.$

Finally, measuring the ancilla qubit wrt the Pauli $Z$ operator, we obtain the desired state with probability $1.$ Indeed, note that $R_y(\theta)=\exp(i\theta\sigma_y/2)=\bmatrix{\cos\frac{\theta}{2} & -\sin\frac{\theta}{2}\\ \sin\frac{\theta}{2} & \cos\frac{\theta}{2}}.$ Then note that $G_{\alpha_j}=R_y(\theta_j),$ where $\theta_j=2 \arctan\left(\sqrt{\frac{d_j}{2|E|-\sum_{l=1}^j d_l}}\right)$ for $j\geq 2,$ and $G_{\alpha_1}=R_y(\theta_1)$ where $\theta_1=2\arctan\left(\sqrt{\frac{d_1}{2|E|-d_1}}\right).$

\section{Estimating subgraph counts}\label{sec:4}

 In this section, we propose methods for estimating subgraph frequency of a graph $G$ using $m$ copies of $\ket{G}$, where $m$ is the number of edges of the subgraph. Our proposal is based on encoding a given subgraph $S$ into a  projector $P_S$ and estimating the subgraph frequency through the estimation of $\bra{G}^{\otimes m}P_S\ket{G}^{\otimes m}$. In particular, we describe the construction of $P_S$ when $S$ is a cycle and a clique for undirected graphs. The operator $P_S$ for other patterns can be constructed similarly, even for directed graphs.

\subsection{Clique counting}

Counting cliques on $m$ vertices in a graph $G$ of $N$ vertices is concerned with finding the number of complete subgraphs $K_m$ of $G.$ In this section, we use ${m\choose 2}$ copies of the embedding state $\ket{G}$ with the aid of $P_{K_m}$ as follows. 

For any pair of $(i,j)\in [N]\times [N]$ with $i\neq j,$ consider the state $\ket{i}\ket{j},$ defined by the $\log_2N$-bit representation of $i$ and $j$, where $[N]=\{0,1,\hdots,N-1\}$. We also define the set $S_m\subset [N]$ with $|\mathcal{S}|=m.$ Let $\mathcal{E}^o$ denote the set of all orderings of all the edges in the clique induced by $S_m.$ 
Note that the number of edges in the clique induced by the set $S_m$ is ${m\choose 2}.$ Now considering $S_m$ as a linearly ordered set, pick all the edges of the clique induced by $S_m$ as the ordered set $\mathcal{E}_{S_m} :=\{(u,v): u<v, \, u,v\in S_m\}.$ With a bit of abuse of notation, we denote $$\ket{\mathcal{E}_{S_m}}=\otimes_{(u,v)\in \mathcal{E}_{S_m}} \ket{u}\ket{v}.$$ For example, if $S_m=\{0,1,2,3\}$ then $\mathcal{E}_{S_m}=\{(0,1), (0,2), (0,3), (1,2), (1,3), (2,3)\}$ and hence $$\ket{\mathcal{E}_{S_m}}=\ket{0}\ket{1} \, \ket{0}\ket{2} \, \ket{0}\ket{3} \, \ket{1}\ket{2} \, \ket{1}\ket{3} \, \ket{2}\ket{3}.$$

Finally, we define 
\begin{equation}\label{eqn:pkm}
    P_{K_m}:= \sum_{S_m\subset [N]} \ket{\mathcal{E}_{S_m}}\bra{\mathcal{E}_{S_m}}.
\end{equation}

Now, the state $\ket{G}^{\binom{m}{2}}$ include terms which associate to any collection of $\binom{m}{2}$ edges in all possible orderings. Indeed,
\begin{eqnarray*}&& \ket{G}^{\binom{m}{2}} =\frac{1}{(\sqrt{2|E|})^{\binom{m}{2}}} \sum_{r_j,c_j\in [N]} a_{r_1c_1}\cdots a_{r_{\binom{m}{2}}c_{\binom{m}{2}}} \\ && \hspace{2cm}\ket{r_1}\ket{c_1} \, \ket{r_2}\ket{c_2} \, \cdots \ket{r_{\binom{m}{2}}}\ket{c_{\binom{m}{2}}}.\end{eqnarray*}

Now measuring the observable $P_{K_m}$ for the quantum state $\ket{G}^{\binom{m}{2}},$ the probability of obtaining $P_{K_m}\ket{G}^{\binom{m}{2}}$ is given by
$$\bra{G}^{\binom{m}{2}} P_{K_m} \ket{G}^{\binom{m}{2}} =\frac{\kappa_m}{(2|E|)^{\binom{m}{2}}},$$ where $\kappa_m$ is the number of cliques on $m$ vertices in $G.$

\subsection{Cycle counting}

Estimating number of cycles of length $k$ i.e. having $k$ edges in the cycle denoted by $C_k$ has been one of the prominent research topics in the literature. Like clique counting, we propose a measurement scheme for the quantum state $\ket{G}^{\otimes k}$ for estimating number of cycles on $k$ vertices.

First, note that a cycle $C_k$ on $k$ vertices has $k$ edges. If $v_1,\hdots,v_k$ are the vertices of $C_k$ then for each ordering of them there is a basis state vector $\ket{v_1}\ket{v_2}\ket{v_2}\ket{v_3}\cdots\ket{v_k}\ket{v_1}$ in $\ket{G}^{\otimes k}.$ Denoting the  vertex set $S_k\subset [N]$ of $C_k$ we define the rank-one project operator
\begin{eqnarray*}
    P_{C_k} = \sum_{S_k} \ket{v_1}\bra{v_1} \, \ket{v_2}\bra{v_2} \, \ket{v_2}\bra{v_2} \, 
    \cdots\ket{v_k}\bra{v_k}\ket{v_1}\bra{v_1}.
\end{eqnarray*} Then we have $$\bra{G}^{\otimes k}P_{C_k}\ket{G}^{\otimes k} = \frac{\mathfrak{c}_k}{2^k\, |E|^k},$$ where $\mathfrak{c}_k$ is the number of cycles on $k$ vertices in $G.$

\subsection{Counting triangles} 

A triangle in a graph is a subgraph which can be considered a a clique or cycle on three vertice i.e. $K_3$ or $C_3.$ Estimating number of triangles has been one of the active areas of research both in quantum and classical computing.
Given a graph $G$ on $N$ vertices with vertex set $V$ and edge set $E$, we consider a register of $6\ceil{\log_2N}$ working qubit register, and six ancillary qubits. Each set of $2\ceil{\log_2N}$-qubit working qubit register along with two ancillary qubits would be used to prepare $3$ copies of $\ket{G}.$ The proposed method for estimating total number of triangles is based on measurement operators to the entire collection of working qubits system.  Thus we have have a system with state \begin{eqnarray*}&& \ket{G}^{\otimes 3}  \frac{1}{(2|E|)^{3/2}}\left(\sum_{\substack{r_1,c_1\in V\\ r_1\neq c_1}}a_{r_1c_1}\ket{r_1}\ket{c_1}\right) \\ && \hfill{\left(\sum_{\substack{r_2,c_2\in V\\ r_2\neq c_2}}a_{r_2c_2}\ket{r_2}\ket{c_2}\right) \left(\sum_{\substack{r_3,c_3\in V\\ r_3\neq c_3}} a_{r_3c_3} \ket{r_3}\ket{c_3}\right)} \\
&=& \frac{1}{(2|E|)^{3/2}} \sum_{\substack{r_1,c_1,r_2,c_2,r_3,c_3\in V\\ r_1\neq c_1, r_2\neq c_2, r_3\neq c_3}} a_{r_1c_1} a_{r_2c_2} a_{r_3c_3} \\ && \hspace{3.5cm}\ket{r_1}\ket{c_1} \ket{r_2}\ket{c_2} \ket{r_3}\ket{c_3}.
\end{eqnarray*}

Now for any triangle $T_{ijk}$ in the graph with vertices $i,j,k$, the corresponding basis state in $\ket{G}^{\otimes 3}$ should be $\ket{i}\ket{j}\ket{j}\ket{k}\ket{k}\ket{i},$ $\ket{j}\ket{i}\ket{i}\ket{k}\ket{k}\ket{j},$ $\ket{k}\ket{i}\ket{i}\ket{j}\ket{j}\ket{k}.$ We set to consider the first one by setting $i<j<k,$ and hence we define the project operator
\begin{eqnarray*}
    P_{\triangle} &=& \sum_{\substack{i,j,k\in V\\i< j< k}} \left(\ket{i}\ket{j}\ket{j}\ket{k}\ket{k}\ket{i}\right) \left(\bra{i}\bra{j}\bra{j}\bra{k}\bra{k}\bra{i}\right) \\
    &=&  \sum_{\substack{i,j,k\in V\\i< j< k}} (\ket{i}\ket{j}\bra{i}\bra{j}) \, (\ket{j}\ket{k}\bra{j}\bra{k}) \, (\ket{k}\ket{i}\bra{k}\bra{i}).
\end{eqnarray*}

Now note that $P_{\triangle}\ket{G}^{\otimes 3}$ would be nonzero only when there is a triangle, and each triangle on three vertices $i,j,k$ would be resulting in three basis states. $\ket{i}\ket{j}\ket{j}\ket{k}\ket{k}\ket{j},$ $\ket{j}\ket{i}\ket{i}\ket{k}\ket{k}\ket{j},$ $\ket{k}\ket{i}\ket{i}\ket{j}\ket{j}\ket{k}.$ Consequently, we have $$P_{\triangle}\ket{G}^{\otimes 3}= \frac{1}{(2|E|)^{3/2}} \sum_{\substack{i,j,k\in V\\i< j< k}} \ket{i}\ket{j}\ket{j}\ket{k}\ket{k}\ket{i},$$
which finally provides the probability of obtaining $P\ket{G}^{\otimes 3}$, given by $$\bra{G}^{\otimes 3}P_{\triangle}\ket{G}^{\otimes 3}=\frac{\kappa_3}{8|E|^3},$$ where $\kappa_3$ is the number of triangles in the graph $G.$ Obviously, $\ket{G}^{\otimes 3}$ is a linear combinations of triangles formed by three vertices $i,j,k$ and the corresponding basis state is given by $\ket{i}\ket{j}\ket{k}.$

\subsection{Sample complexity of estimating $\bra{G}^{\otimes e_s}P_S\ket{G}^{\otimes e_s}$}
 For a given subgraph $S\in\{K_m, C_k\}$, note that the probability of success \begin{equation}\label{eqn:sp}p_s=\bra{G}^{\otimes e_s}P_S\ket{G}^{\otimes e_s}=\frac{\# S}{(2|E|)^{e_s}},\end{equation} where $e_s$ is the number of edges in the subgraph $S,$ and $\# S$ denotes the number of instances of $S$ in the given graph $G.$ There can be multiple ways to estimate $p_s.$ First, we recall the Hoeffding's inequality \cite{hoeffding1963probability} as follows. 

 \begin{proposition}\label{prop:HI}
  Let $X_1, \hdots, X_K$ be iid random variables with values in $\{0,1\}$ and mean $\mathbb{E}[X_j]=p.$ Let $\widehat{p}=\frac{1}{k}\sum_{j=1}^K X_j.$ Then for any $\alpha>0,$
  $$\mbox{Pr}\left(|\widehat{p}-p|\geq \alpha\right)\leq 2\exp\left(-2K\alpha^2\right).$$
 \end{proposition}

We discuss two standard methods for estimating $p_s$: (i) Using Amplitude Estimation \cite{brassard2000quantum,nielsen2001quantum} and (ii) Two-outcome based POVM, and the sample complexities associated with it. 

\subsubsection{Amplitude Amplification and Estimation (AAE)} 

\begin{theorem}\label{thm:aae}
Let $\ket{\psi}$ be a $2\log_2 N$-qubit quantum state that can be prepared by a unitary $U_\Psi$, and let $P_S$ be a projector acting on $(\mathbb{C}^{2^{2\log_2 N}})^{\otimes m}$. Define
$p_s := \bra{\psi}^{\otimes m}P_S\ket{\psi}^{\otimes m}.$
Set $A := U_\psi^{\otimes m}$
and the reflection $S_P:= I - 2P_S.$ Then there exists a quantum algorithm that outputs an estimate $\widetilde{p}$ such that
$|\widetilde{p}-p_s|\le \epsilon$
with success probability at least $1-\delta$, using
$O\!\left(\frac{1}{\epsilon}\log\frac{1}{\delta}\right)$
applications of $A$, $A^\dagger$, and $S_P$.
\end{theorem}

\begin{proof}
Let $A\ket{0}^{\otimes m} = \ket{\psi}^{\otimes m}.$ Then decompose the state with respect to the projector $P_S$ as
$$\ket{\psi}^{\otimes m}
= P_S\ket{\psi}^{\otimes m} + (I-P_S)\ket{\psi}^{\otimes m}.$$

Now define
$$\ket{\psi_G} := \frac{P_S\ket{\psi}^{\otimes m}}{\sqrt{p_s}},
\qquad
\ket{\psi_B} := \frac{(I-P_S)\ket{\psi}^{\otimes m}}{\sqrt{1-p_s}}.$$ Consequently,
$$\ket{\psi}^{\otimes m}
= \sqrt{p_s}\,\ket{\psi_G} + \sqrt{1-p_s}\,\ket{\psi_B},$$
where $\ket{\psi_G}\in \mathrm{Ran}(P_S)$ and $\ket{\psi_B}\in \mathrm{Ran}(I-P_S)$ are orthonormal.

Now let $\theta\in[0,\pi/2]$ be such that $\sin^2\theta = p$ and hence
$$\ket{\psi}^{\otimes m}
= \sin\theta\,\ket{\psi_G} + \cos\theta\,\ket{\psi_B}.$$

Define the reflection about the initial state
$S_0 = I - 2\ket{0}^{\otimes m}\bra{0}^{\otimes m},$
and the Grover diffusion operator
$Q = - A S_0 A^\dagger S_P.$
Further, since
$A S_0 A^\dagger = I - 2\ket{\psi}^{\otimes m}\bra{\psi}^{\otimes m},$
the operator $Q$ is a product of two reflections, one is about $\mathrm{Ran}(P_S)$ and the other one is about $\ket{\psi}^{\otimes m}$. Therefore, $Q$ preserves the two-dimensional subspace
$\mathrm{span}\{\ket{\psi_G},\ket{\psi_B}\}$
and acts on it as a rotation by angle $2\theta$. In particular,
$$Q^j \ket{\psi}^{\otimes m}
=
\sin((2j+1)\theta)\ket{\psi_G}
+
\cos((2j+1)\theta)\ket{\psi_B}.$$

Thus, estimating the eigenphases of $Q$ allows one to estimate $\theta$, and hence $p_s = \sin^2\theta.$
By standard amplitude estimation, one obtains an estimate $\widetilde{\theta}$ such that
$|\widetilde{\theta}-\theta| \le O(\epsilon)$
using $O(1/\epsilon)$ applications of $Q$. Since
$\left|\frac{d}{d\theta}\sin^2\theta\right| = |\sin 2\theta| \le 1,$
it follows that
$$|\widetilde{p}-p|
=
|\sin^2(\widetilde{\theta}) - \sin^2(\theta)|
\le
|\widetilde{\theta}-\theta|
\le \epsilon.$$

Finally, repeating the procedure and taking the median boosts the success probability to at least $1-\delta$, at an additional multiplicative cost of $O(\log(1/\delta))$. Thus the desired result follows.
\end{proof}

 \subsubsection{POVM-based} Given a subgraph structure $S,$ a two-outcome projective measurement associated with a projector $P_S$ is given by $\{P_S,\, I-P_S\}.$ Then we have the following theorem that provides sample complexity of the two-outcome projective measurement method.
 
\begin{theorem}
Let $\ket{\psi}$ and $p_s$ be the same as Theorem \ref{thm:aae}.
Suppose that one performs the two-outcome projective measurement $\{P_S,I-P_S\}$ independently on $T$ copies of $\ket{\psi}^{\otimes m}$, and let $\hat p := \frac{1}{T}\sum_{j=1}^T X_j$
be the empirical frequency of the outcome corresponding to $P_S$, where $X_j\in\{0,1\}$ denotes the outcome of the $j$-th trial. Then for any $\epsilon,\delta>0$, if
$$T \ge \frac{1}{2\epsilon^2}\ln\frac{2}{\delta},$$
it follows that
$\Pr\big(|\hat p-p_s|\ge \epsilon\big)\le \delta.$
 
Consequently, the sample complexity of estimating $p_s$ up to additive error $\epsilon$ with confidence at least $1-\delta$ is
$O\!\left(\frac{1}{\epsilon^2}\log\frac{1}{\delta}\right).$
\end{theorem}

\begin{proof}
Since $\{P_S,I-P_S\}$ is a valid projective measurement, the Born rule implies that for each trial $j$,
$\Pr(X_j=1) = \bra{\psi}^{\otimes m}P\ket{\psi}^{\otimes m} = p_s.$

Thus, $X_1,\dots,X_T$ are i.i.d.\ Bernoulli random variables with mean $p$. The estimator
$\hat p = \frac{1}{T}\sum_{j=1}^T X_j$
is unbiased, i.e., $\mathbb{E}[\hat p]=p_s$.

Since $X_j\in[0,1]$, Hoeffding's inequality gives
$\Pr\big(|\hat p-p|\ge \epsilon\big)\le 2e^{-2T\epsilon^2}.$
Imposing $2e^{-2T\epsilon^2}\le \delta$ yields
$T \ge \frac{1}{2\epsilon^2}\ln\frac{2}{\delta}.$
This completes the proof.
\end{proof}

To compare the two-outcome POVM  method and the AAE-based approach, a key distinction arises in both qubit complexity and classical state-vector simulation cost. In the POVM, each trial requires preparing $\ket{\psi}^{\otimes m}$, which is a $2m\log_2 N$-qubit state, and repeating the experiment $O(1/\epsilon^2)$ times to achieve additive accuracy $\epsilon$. Thus, while the instantaneous qubit requirement is $O(m\log N)$, the total qubit usage scales as $O\!\left(\frac{m\log N}{\epsilon^2}\right)$.

In contrast, AAE requires coherent access to the state preparation unitary and implements controlled reflections, but achieves a quadratic improvement in precision dependence, requiring only $O(1/\epsilon)$ iterations. The instantaneous qubit requirement remains $O(m\log N)$, but the total qubit usage reduces to $O\!\left(\frac{m\log N}{\epsilon}\right)$.

\section{Numerical Simulations: POVM vs AE}

We evaluate the performance of the proposed subgraph-counting framework using both the POVM-based estimator and an amplitude-estimation (AE) approach. Experiments are conducted on Erdős--Rényi graphs $G(N,p_e)$ with varying graph sizes $N$ and edge probabilities $p_e$. For each $p_e$, we generate multiple $300$ graph instances and compute the exact subgraph counts (triangles, $4$-cycles, and $4$-cliques) to obtain the true success probability $p$ as in Eqn. (\ref{eqn:sp})
where $m$ denotes the number of edges in the subgraph. The estimated count is obtained by scaling $\hat p$ appropriately.

The POVM-based estimator is simulated via repeated two-outcome measurements, modeled as Bernoulli trials with success probability $p$. For a target additive error $\epsilon$ and failure probability $\delta$, the number of measurement shots is set to $T=\lceil (2\epsilon^2)^{-1}\ln(2/\delta)\rceil$. In contrast, amplitude estimation is simulated at the oracle level by exploiting the parameterization $p=\sin^2\theta$ and introducing a perturbation of order $O(1/M)$ in $\theta$, where $M=O(\epsilon^{-1}\log(1/\delta))$ denotes the query complexity. This captures the characteristic quadratic improvement of AE over sampling-based methods without explicitly simulating the full quantum circuit.

We compare the estimators using the normalized root mean squared error (RMSE) in Fig.~\ref{fig:fig_TC_rmse}, Fig.~\ref{fig:4cC_rmse}, Fig.~\ref{fig:4KC_rmse} of the recovered subgraph counts for triangles, $4$-cycles and $4$-cliques respectively, where 
$$\mathrm{Normalized\;RMSE} =\frac{\sqrt{\mathbb{E}\big[(\widehat{X}-X)^2\big]}}{\mathbb{E}[X]},$$
with $X$ denotes the true scaled quantity which is obtained by known classical algorithms in the simulations and $\widehat{X}$ its proposed estimators. Across all experiments, AE consistently achieves lower normalized RMSE than the POVM-based method. This improvement is explained by the difference in error scaling: the POVM estimator exhibits variance $\mathrm{Var}(\hat p)=p(1-p)/T$, leading to an error of order $O(1/\sqrt{T})$, whereas AE achieves $O(1/M)$ error in estimating $p$. Consequently, for comparable resource budgets, AE provides a quadratic advantage in precision.

The performance gap is more pronounced for higher-order motifs such as 4-cycles and 4-cliques, where the underlying success probability $p$ becomes small. In these regimes, the POVM estimator operates in a noise-dominated setting with normalized RMSE approaching $1$, while AE maintains improved accuracy due to its coherent amplification mechanism. These results highlight the advantage of amplitude estimation for rare-event estimation in quantum subgraph counting, while also demonstrating that the POVM-based method provides a simple and scalable baseline with quantum logarithmic space requirements. The values of $(\epsilon,\delta)$  control the amplification of estimation errors. Since the subgraph count is obtained by scaling $\hat p$ by a factor of $(2E)^m$, even small deviations between $\hat p$ and the true $p$ can lead to significant errors in the estimated count, thereby increasing the normalized RMSE.

\begin{figure}[htbp]
    \centering
    {{\includegraphics[width=0.45\textwidth]{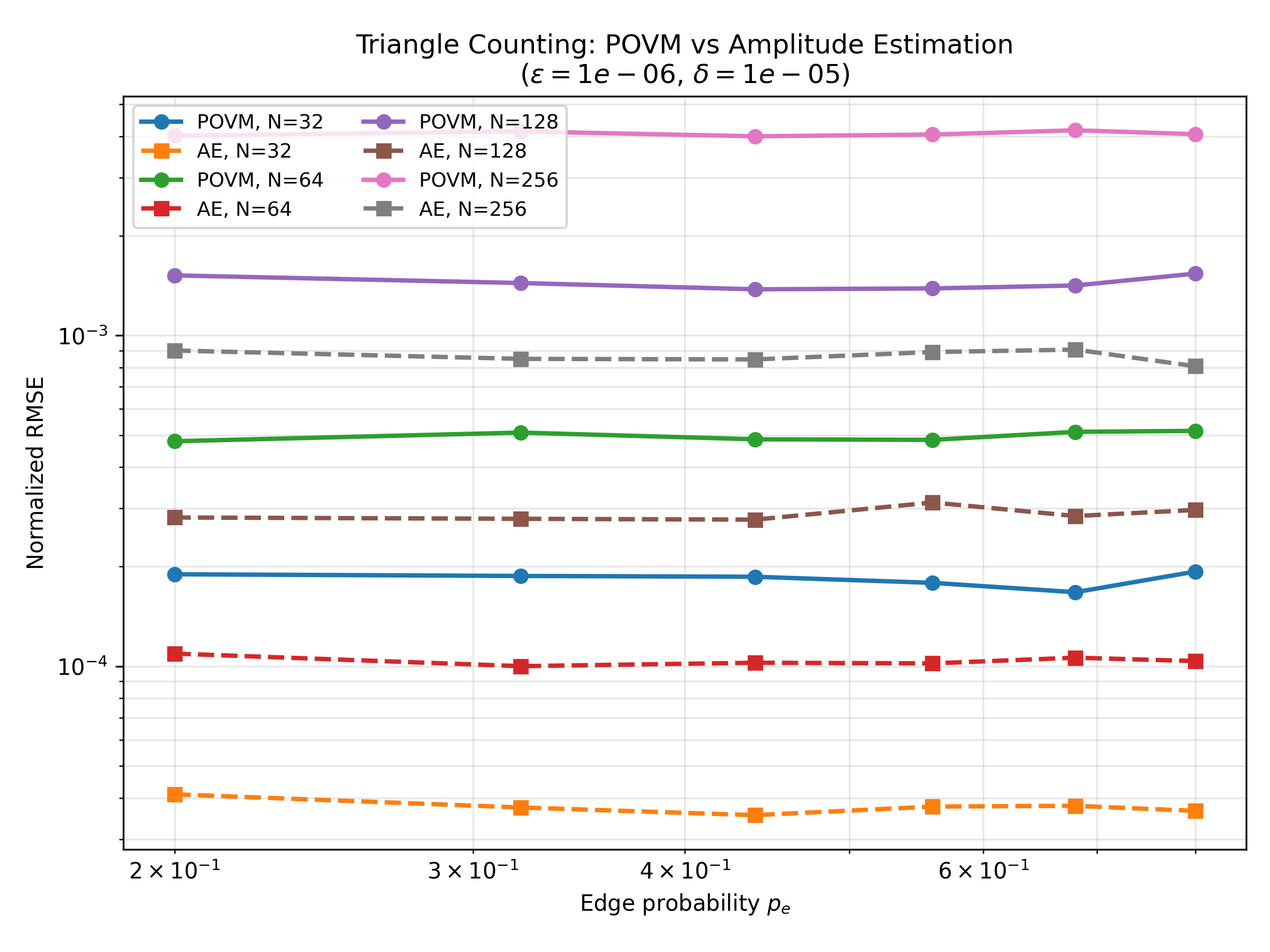}}}
    \caption{Normalized RMSE for estimating triangle counts in ER graphs on $N$ vertices }
    \label{fig:fig_TC_rmse}
\end{figure}

\begin{figure}[htbp]
    \centering
    {{\includegraphics[width=0.45\textwidth]{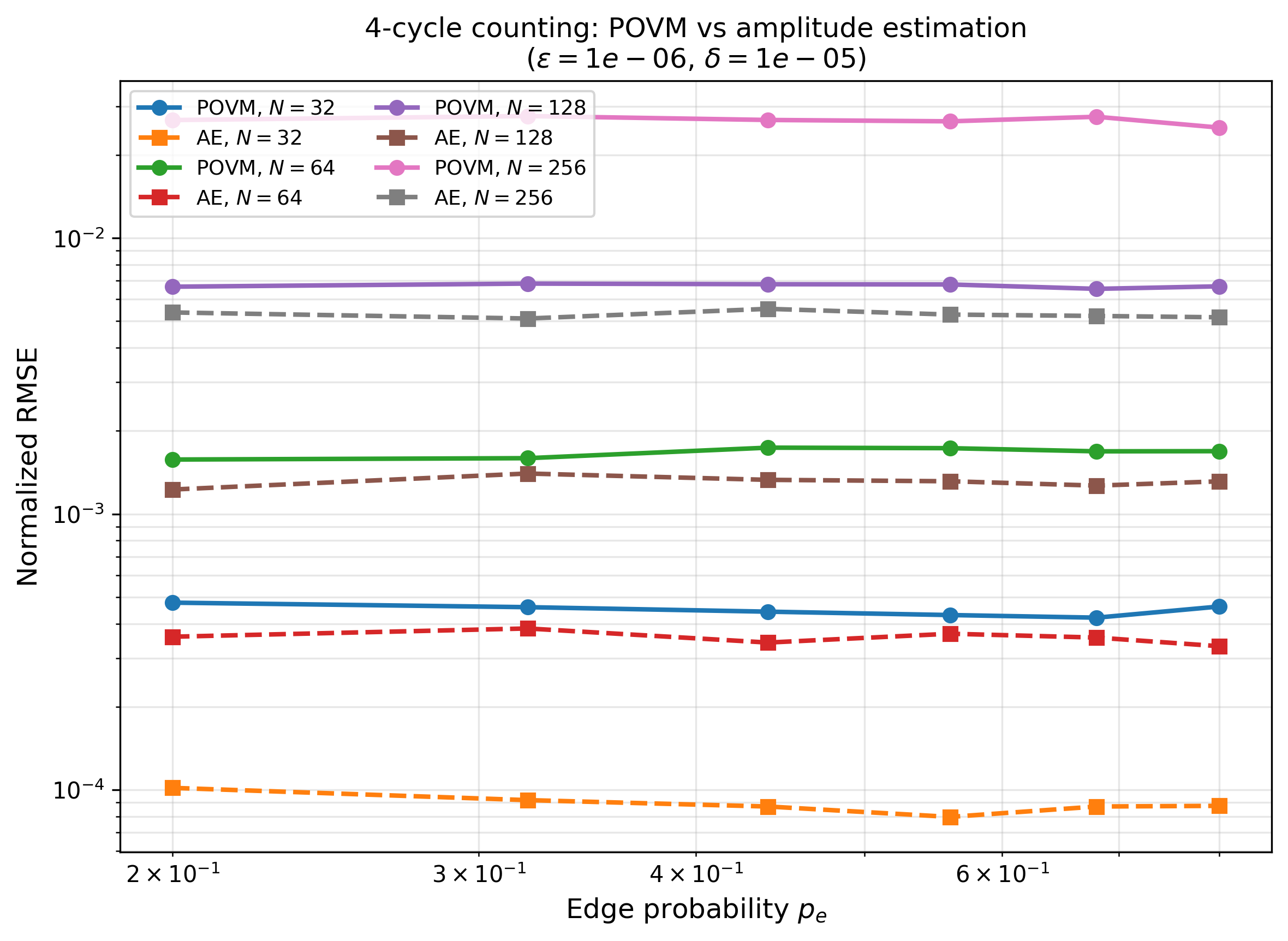}}}
    \caption{Normalized RMSE for estimating $4$-cycle counts in ER graphs on $N$ vertices}
    \label{fig:4cC_rmse}
\end{figure}

\begin{figure}[htbp]
    \centering
    {{\includegraphics[width=0.45\textwidth]{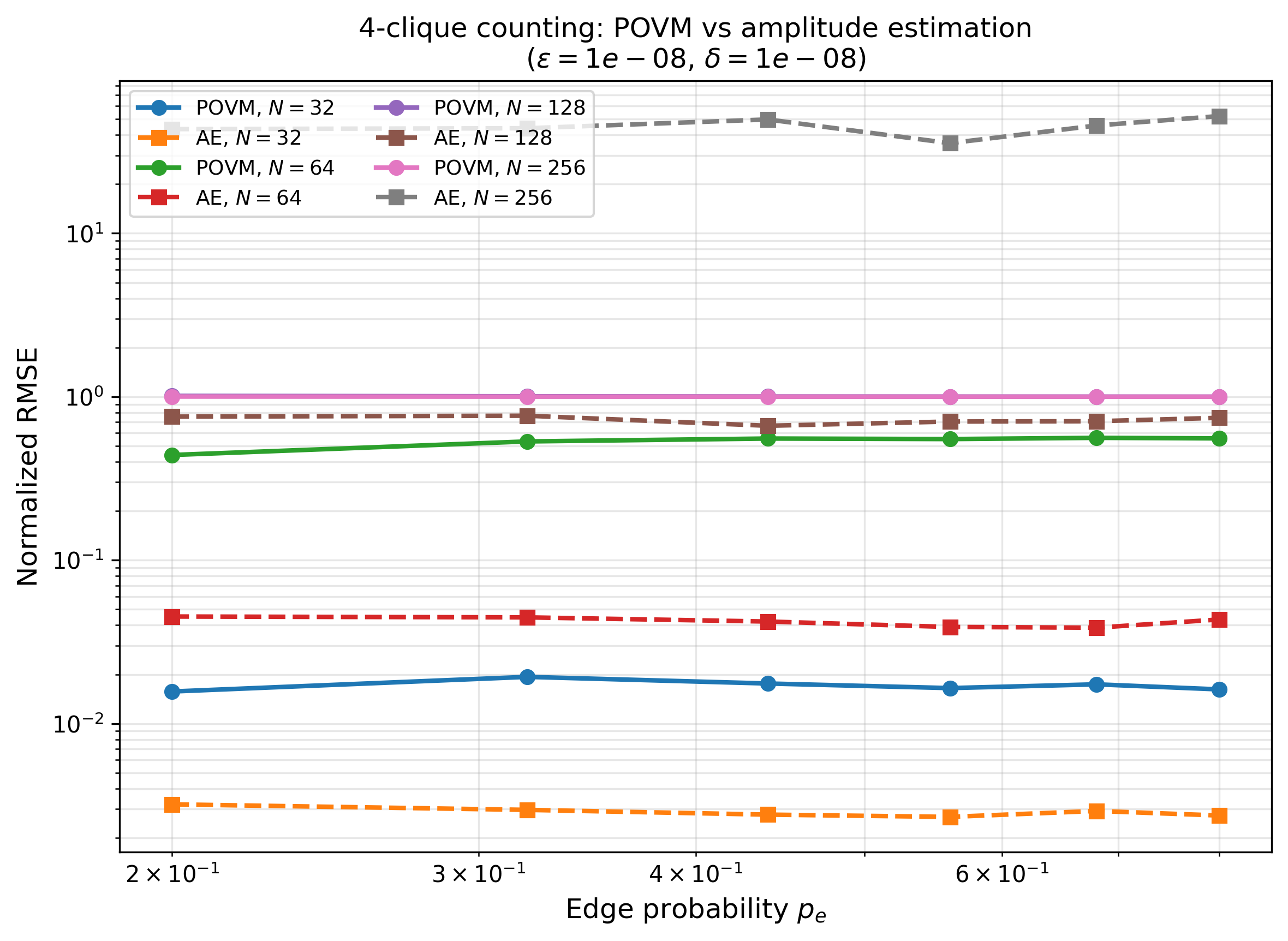}}}
    \caption{Normalized RMSE for estimating $4$-clique counts in ER graphs on $N$ vertices}
    \label{fig:4KC_rmse}
\end{figure}

\section{Conclusion} 

In this work, we develop quantum circuit models for preparing a graph adjacency state that encodes the adjacency relations of a graph as a normalized vectorization of its adjacency matrix. The construction applies to both directed and undirected graphs through appropriate normalization, providing a unified framework for quantum graph embedding. We leverage this state to estimate subgraph frequencies, including cliques and cycles, by performing measurements on multiple copies of the adjacency state, where the number of copies scales with the number of edges in the subgraph. The resulting algorithms for motif counting operate in quantum logspace, and no such algorithms are known in classical computation. An important direction for future work is to carry out circuit simulations of the proposed algorithms and assess their performance in the presence of noise on quantum hardware. \\

\noin{\bf Acknowledgment.} The author thanks Hannes Leipold, Sarvagya Upadhyay, Hirotaka Oshima, and
Yasuhiro Endo for their insightful comments and discussion. 

\bibliographystyle{plain}
\bibliography{reff}
\end{document}